\documentclass[sigconf,nonacm]{acmart}

\usepackage{graphicx}
\usepackage{amsmath}
\usepackage{amssymb}
\usepackage{amsfonts}
\usepackage{mathtools}
\usepackage{paralist}
\usepackage{xspace}
\usepackage{url}
\usepackage{comment}
\usepackage[ruled,linesnumbered,noend]{algorithm2e}
\usepackage{bbold}
\usepackage{enumerate}
\usepackage{xfrac}
\usepackage{booktabs}
\usepackage{tabularx}
\usepackage{makecell}
\usepackage{subcaption}

\usepackage{tikz}
\usetikzlibrary{arrows,shapes,snakes,automata,backgrounds,petri}
\usetikzlibrary{positioning}

\usepackage{amsthm}
\newtheoremstyle{mystyle}
  {}
  {}
  {\itshape}
  {}
  {\bfseries}
  {.}
  { }
  {\thmname{#1}\thmnumber{ #2}\thmnote{ (#3)}}
\theoremstyle{mystyle}
\newtheorem{theorem}{Theorem}
\newtheorem{lemma}[theorem]{Lemma}
\newtheorem{corollary}[theorem]{Corollary}
\newtheorem{definition}{Definition}

\newcommand{\past}{\textit{past}}

\newcommand{\fee}{\textit{fee}}

\newtheorem*{definition*}{Definition}

\usepackage{color}
\newcount\Comments  
\newcommand{\mytodo}[2]
{\ifnum\Comments=1
	{\marginpar{\small{\color{#1}	#2}}}\fi}

\makeatletter
\newcommand{\nosemic}{\renewcommand{\@endalgocfline}{\relax}}
\newcommand{\dosemic}{\renewcommand{\@endalgocfline}{\algocf@endline}}
\let\oldnl\nl
\newcommand{\nonl}{\renewcommand{\nl}{\let\nl\oldnl}}
\makeatother

\title{ABC: Proof-of-Stake without Consensus}
\author{Jakub Sliwinski and Roger Wattenhofer}
\email{{jsliwinski,wattenhofer}@ethz.ch}
\affiliation{\institution{ETH Zurich}}

\begin{document}

\begin{abstract}
We introduce a new permissionless blockchain architecture called ABC. ABC is completely asynchronous, and does rely on neither randomness nor proof-of-work. ABC can be parallelized, and transactions have finality within one round trip of communication. However, ABC satisfies only a relaxed form of consensus by introducing a weaker termination property. 
Without full consensus, ABC cannot support certain applications, in particular ABC cannot support general smart contracts. However, many important applications do not need general smart contracts, and ABC is a better solution for these applications. In particular, ABC can implement the functionality of a cryptocurrency like Bitcoin, replacing Bitcoin's energy-hungry proof-of-work with a proof-of-stake validation.
\end{abstract}

\maketitle

\section{Introduction}\label{sec:intro}

Nakamoto's Bitcoin protocol \cite{nakamoto2008bitcoin} has taught the world how to achieve trust without a designated trusted party. The Bitcoin architecture provides an interesting deviation from classic distributed systems approaches, for instance by using proof-of-work to allow anonymous participants to join and leave the system at any point, without permission.

However, Bitcoin's proof-of-work solution comes at serious costs and compromises. The security of the system is directly related to the amount of investments in designated proof-of-work hardware, and to spending energy to run that hardware. Since the system's participants that provide the distributed infrastructure (often called miners or validators) bear significant costs (hardware, energy), the protocol compensates them with Bitcoins. However, adversaries might disrupt this scheme by bribing the miners to behave untruthfully or disrupt the reward payments.

To make matters worse, proof-of-work protocols assume critical requirements related to the communication between the participants regarding message loss and timing guarantees. In other words, such protocols are vulnerable to attacks on the underlying network.

In the decade since the original Bitcoin publication, researchers have tried to address the wastefulness of proof-of-work. One of the most prominent research directions is replacing Bitcoin's proof-of-work with a proof-of-stake approach. In proof-of-stake designs, miners are replaced with participants who contribute to running the system according to the amounts of cryptocurrency they hold. Alas, proof-of-stake protocols require similar communication guarantees as proof-of-work, and thus can also be attacked by disrupting the network. Moreover, proof-of-stake introduces some of its own problems. Prominently, existing proof-of-stake designs critically rely on randomness. To achieve consensus, the participants of such systems repeatedly choose a leader among themselves. Despite being random, this choice needs to be taken collectively and in a verifiable way, which complicates the problem.

Due to the way blockchains typically process transactions, participants have to wait significant amount of time before they can be confident that their transactions are accepted by the system. For example, it usually takes around an hour for merchants to accept Bitcoin transactions as confirmed, which is unacceptable for time-sensitive applications.

In his seminal paper, Nakamoto made the crucial assumption that his system has to be able to totally order the transactions submitted to the system in order to reject the fraudulent ones. However, meeting this requirement is equivalent to solving the problem known as consensus. Nakamoto's assumption has shaped the design of blockchain systems to this day. Thus, many blockchain systems achieve consensus while not taking advantage of this powerful property, but suffering the associated costs.

\paragraph{\bf Our Contribution}\label{sec:contrib}
\renewcommand{\arraystretch}{1.5}
\begin{table*}\caption{Comparison of ABC to selected BFT/blockchain protocols. Permissioned protocols are on the left, permissionless protocols on the right. We mark all protocols providing full consensus as supporting general smart contracts, even though particular implementations might not feature smart contracts.}\label{tbl:comparison}
\centering
\begin{tabularx}{\textwidth}{|c| >{\centering\arraybackslash}X| >{\centering\arraybackslash}X| >{\centering\arraybackslash}X !{\vrule width 1pt} >{\centering\arraybackslash}X| >{\centering\arraybackslash}X| >{\centering\arraybackslash}X| >{\centering\arraybackslash}X| } \toprule
& PBFT\cite{castro1999practical} & HoneyBadger BFT\cite{miller2016honey} & Broadcast-based\cite{gupta2016non} & Bitcoin and Ethereum\cite{wood2014ethereum} & Ouroboros\cite{kiayias2017ouroboros} & Algorand\cite{gilad2017algorand} & ABC \\ \midrule
Permissionless & & & & \checkmark &\checkmark &\checkmark & \checkmark \\ \hline
\makecell{Proof-of-work \\ free} & \checkmark &\checkmark &\checkmark & & \checkmark &\checkmark &\checkmark \\ \hline
Finality & \checkmark & \checkmark &\checkmark & & & \checkmark &\checkmark \\ \hline
Asynchronous &  & \checkmark & \checkmark & & & & \checkmark \\ \hline
Deterministic & \checkmark & & \checkmark & & & & \checkmark \\ \hline
Parallelizable & & & \checkmark & & & & \checkmark \\ \hline
\makecell{General smart \\ contracts} &\checkmark &\checkmark & & \checkmark &\checkmark & \checkmark & \\ \bottomrule
\end{tabularx}

\end{table*}

We relax the usual notion of consensus to extract the requirements necessary for an efficient cryptocurrency. Thus we introduce an asynchronous blockchain design that features an array of advantages compared to alternatives. In other words, we present an Asynchronous Blockchain without Consensus (ABC). ABC offers a host of exciting properties:  

\begin{description}
\item[Permissionless:] Most importantly, ABC offers its advantages without relying on permissioned participation. ABC is permissionless in the same way as other proof-of-stake systems, where participants of the system freely exchange cryptocurrency tokens. Token holders run the system by verifying new transactions. Additionally, any token holder can indicate any other participant to take his part in this process, but preserving his ownership of the associated tokens.

\item[Asynchronous:] ABC does not require the messages to be delivered within any known period of time. Thus ABC is fully resilient to all network-related threats, such as delaying messages, denial-of-service or network eclipse attacks. An adversary having complete control of the network obviously can delay progress of the system (by simply disabling communication), but otherwise cannot interfere with the protocol or trick the participants in any way. Previously approved transactions cannot be invalidated and impermissible transactions cannot be approved. 

\item[Parallelizable:] In ABC, validators running the system can parallelize the processing of transactions. There is no limit to the number of transactions a validator can process by parallelization.

\item[Final:] Under normally functioning network communication, transactions in ABC are instantly confirmed. Confirmation is final and impossible to revert. This is in stark contrast to systems such as Bitcoin, where the confidence in a transaction being confirmed only probabilistically increases with the passage of time.

\item[Deterministic:] We assume the functionality provided by asymmetric encryption and hashing. Apart from these cryptographic necessities, ABC is completely deterministic and surprisingly simple.

\item[Proof-of-stake:] Unlike proof-of-work, the security of the system does not depend on the amount of devoted resources such as energy, computational power, memory, etc. Instead, similarly to other proof-of-stake protocols, ABC requires that more than two thirds of the system's cryptocurrency is held by honest participants.

\end{description}

However, ABC does not support consensus. This prevents ABC from supporting applications that involve smart contracts open for interaction with anybody. For example, the smart contract functionality of Ethereum cannot be directly implemented with ABC.
Many important applications (e.g., cryptocurrencies or IoT systems), do not require consensus, and ABC offers an advantageous solution for these applications.

Table \ref{tbl:comparison} compares the properties of ABC with some of the most relevant existing BFT/blockchain paradigms. Many more protocols exist that improve some aspects, for example many protocols improve upon PBFT. While many of these protocols are more scalable and efficient than the original PBFT, they share the fundamental disadvantages of PBFT: They are not permissionless, they are not parallelizable, and in order to make progress (``liveness''), they need synchronous communication.

Table \ref{tbl:comparison} shows the close relation of ABC with broadcast-based protocols. One may argue that ABC brings the simplicity, robustness and efficiency of broadcast-based protocols to the permissionless world.

\section{Relaxing Consensus}\label{sec:relaxing-consensus}

A cryptocurrency needs to be resilient to some of the participants of the system (called agents) behaving maliciously. The problem of establishing consensus in such an environment is also known as Byzantine agreement. The agents behaving truthfully are called honest, and malicious agents are called Byzantine.

In the context of a cryptocurrency, some form of consensus is used to solve the problem of double-spending: Suppose Alice holds one cryptocurrency coin. Now Alice sets up a transaction that transfers her coin to Bob (in exchange for a good or service). However, Alice wants to cheat, trying to simultaneously spend the same coin in another transaction to Carol. Upon receiving one (or both) of Alice's transactions, honest agents need to agree on what happens to Alice's coin, preventing Alice from doubling her money. In this context, according to the usual definition, achieving consensus consists of the following requirements:
\begin{definition*}[Consensus]\label{def:consensus}
\mbox{}
\begin{description}
    \item[Agreement:] If some honest agent accepts a transaction, every honest agent will accept the same transaction. No two conflicting transactions are  accepted.
    \item[Validity:] If every honest agent observes the same transaction (there are no conflicting transactions), this transaction is accepted by honest agents.
    \item[Termination:] Every honest agent accepts one of the transactions. If messages are delivered quickly, the consensus protocol terminates quickly.
\end{description}
\end{definition*}

The key insight leading to the relaxation of this definition, is that cheaters do not need to enjoy any guarantees. Alice from above tried to cheat by issuing (cryptographically signing) two conflicting transactions. ABC will not need to guarantee that any of Alices's two conflicting transactions will be accepted. In fact, if she tries to cheat, Alice might lose her coin.

On the other hand, an honest Alice will only create one transaction spending her coin. Thus, every honest agent will see the same transaction. Hence we can relax consensus to guarantee termination only for honest agents:

\begin{definition*}[ABC Consensus]\label{def:abcconsensus}
\mbox{}
\begin{description}
    \item[Agreement:] Same as above.
    \item[Validity:] Same as above.
    \item[Honest-Termination] If every honest agent observes the same transaction (and no conflicting transactions), this transaction is accepted by all honest agents. If messages are delivered quickly, the consensus protocol terminates quickly.
\end{description}
\end{definition*}

Under this relaxed notion of consensus, if Alice tries to cheat, it is possible that neither Bob nor Carol will accept Alice's transaction. Some honest agents might see one of the transactions first, while others might see the other first. Then the requirement of Honest-Termination does not apply, and the transactions might stay without a resolution forever. This turn of events can be seen as Alice losing her coin due to misbehaviour.

Otherwise Consensus and ABC Consensus do not differ. Agreed upon results are final, conflicting results are precluded and honest transactions are accepted quickly.

Surprisingly, despite the difference being so insignificant with respect to the functioning of a cryptocurrency, this relaxation allows ABC to combine a large set of advantages.

\section{Intuition}





For simplicity of presentation, we describe ABC in the terminology of a cryptocurrency. We refer to the cryptocurrency managed by the protocol as the money. A more formal description follows in Section \ref{sec:protocol}. 

\paragraph{\bf Transactions.}

As usual in cryptocurrencies, the main operation is a transaction, which transfers money from one or more inputs to one or more outputs. Inputs and outputs are money amounts paired with keys required to spend them. Every transaction refers to at least one previous transaction, such that all transactions form a directed acyclic graph (DAG).

\paragraph{\bf Validators.}

In proof-of-stake systems, the agents that own some of the money in the system also run the system. These agents are staying online and participating in validating transactions. In ABC's design, we do not require agents to stay online and participate, but allow agents to delegate this responsibility to other agents. All agents that stay online and validate transactions are called validators. Every agent can choose to be a validator. If an agent does not want to be validator, the agent can indicate (with an additional public key in a transaction) who should be the validator to whom the stake is delegated.
Validators stay online and sign correct transactions. A transaction is correct if the validator did not see another transaction spending the same input(s). The system works correctly as long as agents holding more than two-thirds of the system's money indicate honest validators.

\paragraph{\bf Confirmations.}

A transaction $t$ is confirmed by the system if enough validators ack (acknowledge by signing) $t$. An ack cannot be revoked. If a transaction receives enough acks, no other transaction conflicting with $t$ can become confirmed. In particular, if a cheating Alice attempts to issue a transaction $t'$ which is trying to spend (some of) the same input(s) as $t$, the system will never confirm $t'$. If Alice issues two conflicting transactions $t$ and $t'$ at roughly the same time, it is possible that (a) either $t$ or $t'$ gets confirmed (but not both), or (b) neither $t$ nor $t'$ are ever confirmed. Case (b) happens if some (but not enough) validators see and sign $t$, while others see and sign $t'$. The system might stay in this state forever with the validators' approval split between $t$ and $t'$, with no clear majority. Such a situation can only arise if Alice tried to cheat. The result is equivalent to the misbehaving agent losing the money she attempted to double-spend, and does not constitute any threat to the system.

It is somewhat intuitive to verify that such a system does work correctly if the validating power amounts are statically assigned to the validators, and a set of validators controlling more than two-thirds of the cryptocurrency obeys the protocol. In Section \ref{sec:double-spending} we will show that our system still works correctly when the agents can freely exchange the cryptocurrency and change the appointed validators, even in the harsh conditions of an asynchronous network. Thus, we establish a system with the participation model similar to proof-of-stake protocols, but much simpler than known proof-of-stake protocols.


\section{Model}

\paragraph{\bf Agents and Adversary.}

Our blockchain is used and maintained by its participants called agents. Agents who follow the protocol are called honest. The set of agents who do not follow the protocol is controlled by the adversary. The adversary behaves in an arbitrary (Byzantine) way.

Similarly to other proof-of-stake systems, we assume that at any time, the adversary owns less than one-third of the cryptocurrency present in the system. We introduce more concepts in order to state this requirement precisely in Section \ref{sec:adversary-stake}.

\paragraph{\bf Asynchronous Communication.}
All agents are connected by a virtual network similar to Bitcoin's, where agents can broadcast their messages to all other agents. Like in Bitcoin, new agents can join this network to receive new and prior messages. Agents can also leave the network.

The network is asynchronous: The adversary controls the network, dictating when messages are delivered and in what order. Messages are only required to reach the recipient {\it eventually}, without any bound on the time it might take. Under such weak network requirements, an adversary delaying the delivery of messages can delay the progress of an agent, but otherwise will not be able to interfere with the protocol or trick honest agents.

In the Appendix, we summarize why many of the protocols in Table \ref{tbl:comparison} cannot be considered asynchronous.

\paragraph{\bf Cryptographic Primitives.}

We assume the functionality of asymmetric encryption where a public key allows every agent to verify a signature of the associated secret key. Agents can freely generate public/secret key pairs.

We also assume cryptographic hashing, where for every message a succinct, unique hash can be computed. Whenever we mention references between transactions in our protocol, we mean hashes that uniquely identify the referenced data.

Apart from these two primitives, the ABC protocol is completely deterministic.

\section{Protocol}
\label{sec:protocol}

\subsection{Transactions}

\paragraph{\bf Outputs.}

Outputs are the basic unit of information. Outputs are included in transactions to identify money holders and corresponding money amounts. 

\begin{definition}[Output]
An output contains:
\begin{itemize}
    \item {\it Value:} A number representing the amount of money.
    \item {\it Owner key:} A public key. The agent holding the associated secret key is the owner of the money.
\end{itemize}
\end{definition}

In general, agents could reuse their keys for multiple outputs, but for simplicity of presentation we assume that the owner key always uniquely identifies a single output.

\paragraph{\bf Transactions.}

A transaction is a request issued by an agent (or a set of agents) to transfer money to other agent(s). Outputs of a transaction identify recipients of the transaction. The transaction also indicates a validator -- some agent devoted to maintaining the system.

For simplicity of presentation we assume that outputs uniquely identify the originating transaction.

\begin{definition}[Transaction]
A transaction $t$ contains:
\begin{itemize}
	\item A set of inputs, where each input is an output of some previous transaction. Transaction $t$ is said to spend these inputs.
    \item A set of outputs. The sum of values of the outputs equals the sum of values of the inputs.
    \item {\it Validator key:} A public key. The agent holding the associated secret key is indicated as the validator.
\end{itemize}

The sum of values of the outputs is called the value of $t$.
The transaction is signed by all secret keys associated with the inputs.
\end{definition}

\paragraph{\bf Genesis.}

The {\it genesis} is a special transaction without inputs (and hence without the associated signatures). The genesis is hard-coded in the protocol and known upfront to every agent. The genesis describes the initial distribution of money among the original agents and the initial validators (which could be the same as the original agents).

The values of all genesis outputs sum up to $M$, so $M$ is the total money in the system.

\tikzstyle{tx} = [draw=black, fill=gray!5, thick,
    rectangle, rounded corners, inner sep=4pt, inner ysep=4pt, text centered]
\tikzstyle{tx_conf} = [draw=black, fill=blue!20, thick,
    rectangle, rounded corners, inner sep=4pt, inner ysep=4pt, text centered]
\tikzstyle{tx_igno} = [draw=black, fill=gray!80, thick,
    rectangle, rounded corners, inner sep=4pt, inner ysep=4pt, text centered]

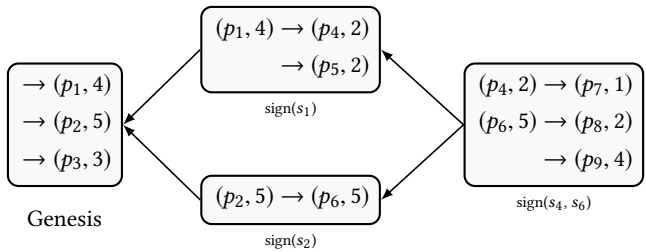
\begin{figure}
\centering
\begin{tikzpicture}

\node[tx](GENESIS) 
 {$\begin{aligned}
	\to (p_1, 4)\\
	\to (p_2, 5)\\
	\to (p_3, 3)
	\end{aligned}$
  };
 
\node[below of=GENESIS,
  yshift=-0.25cm,
  anchor=center]
  {Genesis};

\node[tx,
right of=GENESIS,
yshift=1cm,
xshift=2cm,
anchor=center,
label=below:{\scriptsize sign($s_1$)}
](T1) 
 {$\begin{aligned}
       (p_1, 4) &\to (p_4, 2)\\
       & \to (p_5, 2)
   \end{aligned}$};

\draw (T1.west) edge[-latex, semithick] (GENESIS.east);

\node[tx,
right of=GENESIS,
yshift=-1cm,
xshift=2cm,
anchor=center,
label=below:{\scriptsize sign($s_2$)}
](T2)
 {$\begin{aligned}
       (p_2, 5) \to (p_6, 5)
   \end{aligned}$};

\draw (T2.west) edge[-latex, semithick] (GENESIS.east);

\node[tx,
right of=T2,
yshift=1cm,
xshift=2.5cm,
anchor=center,
label=below:{\scriptsize sign($s_4$, $s_6$)}
](T3)
 {$\begin{aligned}
       (p_4, 2) &\to (p_7, 1)\\
       (p_6, 5) &\to (p_8, 2)\\
       &\to (p_9, 4)
   \end{aligned}$};

\draw (T3.west) edge[-latex, semithick] (T1.east);
\draw (T3.west) edge[-latex, semithick] (T2.east);

\end{tikzpicture}

\caption{Example DAG of transactions, validator keys are omitted. The $p_i$'s are owner keys, and $s_i$'s are the corresponding secret keys.}
\label{fig:ex1}
\end{figure}

\subsection{\bf Validators}
\label{sec:validators}

Intuitively, validators are agents devoted to maintaining the system. Validators listen for transactions being broadcast, and sign them if they are not misbehaving. After a transaction $t$ with a value of $m$ is processed by the system, the ``signing power'' of the validator $v$ indicated in transaction $t$ increases by $m$ (at the cost of the validators indicated in transactions that output the inputs of $t$). To spend an output of $t$, the owner of an output must later broadcast a new transaction, as $v$ does not control how the outputs of $t$ will be spent. An owner of an output of $t$ can change the appointed validator $v$ to any other validator by spending $t$'s output (for instance by self-sending the money), including a different validator key. Any agent can also indicate herself as the validator.

The validator $v$ signs transactions in the system to contribute to their confirmation, and the contribution is proportional to the amount of money delegated to $v$.

\paragraph{\bf Efficiency} In  ABC, we expect the number of validators to be naturally relatively small, such that a small number of validator's signatures will be enough to confirm a transaction. We think of a validator in ABC to be the equivalent of mining pools in Bitcoin. The set of validators in ABC will shift and change over time, similarly to Bitcoin mining pools.

Without mitigation, an excessive number of (relatively powerless) validators would require excessive numbers of signatures to be broadcast in the network. If one worries about the number of validators being too large, validators can be incentivized (or required) to form groups by the protocol. For example, the protocol might state that too small validators receive smaller fees from transaction confirmation. We will discuss this in more detail in Section \ref{sec:txfees}.

\subsection{Confirmations}

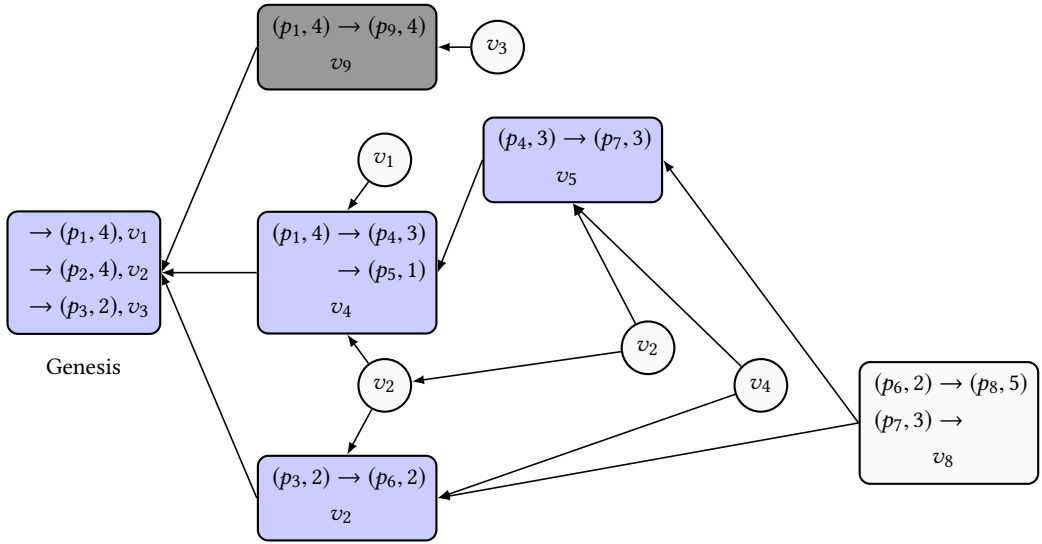
\begin{figure*}
\centering
\begin{tikzpicture}

\node[tx_conf](GENESIS)
 {$\begin{aligned}
   &\to (p_1, 4), v_1\\
   &\to (p_2, 4), v_2\\
   &\to (p_3, 2), v_3
   \end{aligned}$};
 
\node[below of=GENESIS,
  yshift=-0.25cm,
  anchor=center]
  {Genesis};

\node[tx_conf,
right of=GENESIS,
yshift=0cm,
xshift=2.5cm,
anchor=center
](T1)
 {$\begin{aligned}
       (p_1, 4) &\to (p_4, 3)\\
       & \to (p_5, 1)\\
       & v_4
   \end{aligned}$};

\draw (T1.west) edge[-latex, semithick] (GENESIS.east);

\node[tx_conf,
right of=T1,
yshift=1.5cm,
xshift=2cm,
anchor=center
](T5)
 {$\begin{aligned}
       (p_4, 3) &\to (p_7, 3)\\
       & v_5
   \end{aligned}$};

\draw (T5.west) edge[-latex, semithick] (T1.east);

\node[tx_conf,
right of=GENESIS,
yshift=-3cm,
xshift=2.5cm,
anchor=center
](T2)
 {$\begin{aligned}
       (p_3, 2) &\to (p_6, 2)\\
       & v_2
   \end{aligned}$};

\draw (T2.west) edge[-latex, semithick] (GENESIS.east);

\node[tx_igno,
right of=GENESIS,
yshift=3cm,
xshift=2.5cm,
anchor=center
](T4)
 {$\begin{aligned}
       (p_1, 4) &\to (p_9, 4)\\
       & v_9
   \end{aligned}$};

\draw (T4.west) edge[-latex, semithick] (GENESIS.east);

\node[tx,
right of=T5,
yshift=-3.5cm,
xshift=4cm,
anchor=center
](T3)
 {$\begin{aligned}
       (p_6, 2) &\to (p_8, 5)\\
       (p_7, 3) &\to\\
       & v_8
   \end{aligned}$};

\draw (T3.west) edge[-latex, semithick] (T5.east);
\draw (T3.west) edge[-latex, semithick] (T2.east);

\node[right of=GENESIS,
yshift=1.5cm,
xshift=3cm,
anchor=center,
draw=black, minimum size=7mm, inner sep=1pt, fill=gray!5, thick, circle](V1){$v_1$};

\draw (V1) edge[-latex, semithick] (T1.north);

\node[right of=GENESIS,
yshift=-1.5cm,
xshift=3cm,
anchor=center,
draw=black, minimum size=7mm, inner sep=1pt, fill=gray!5, thick, circle](V2){$v_2$};

\draw (V2) edge[-latex, semithick] (T1.south);
\draw (V2) edge[-latex, semithick] (T2.north);

\node[right of=V1,
yshift=-3cm,
xshift=4cm,
anchor=center,
draw=black, minimum size=7mm, inner sep=1pt, fill=gray!5, thick, circle](V4){$v_4$};

\draw (V4) edge[-latex, semithick] (T2.east);
\draw (V4) edge[-latex, semithick] (T5.south);

\node[right of=V2,
yshift=0.5cm,
xshift=2.5cm,
anchor=center,
draw=black, minimum size=7mm, inner sep=1pt, fill=gray!5, thick, circle](V3){$v_2$};

\draw (V3) edge[-latex, semithick] (T5.south);
\draw (V3) edge[-latex, semithick] (V2);

\node[right of=T4,
yshift=0cm,
xshift=1cm,
anchor=center,
draw=black, minimum size=7mm, inner sep=1pt, fill=gray!5, thick, circle](V33){$v_3$};

\draw (V33) edge[-latex, semithick] (T4.east);

\end{tikzpicture}

\caption{Example transaction DAG, $p_i$'s represent the owners and $v_i$'s the validators. Circle nodes are acks labelled by the issuing validators. Acks point to the transactions being signed or the previous acks of the same validator. Blue transactions are confirmed based on the acks. When issuing an ack, validators have to point to the previously issued ack, as exhibited by $v_2$. The gray transaction is an attempt at double-spending; it conflicts with a confirmed transaction and will never be confirmed.}

\label{fig:example_dag}
\end{figure*}

A validator broadcasts a message called an {\it ack} to communicate the new set of transactions it signed.

\begin{definition}[ack]
An ack contains:
\begin{itemize}
\item A reference to the previous ack issued by the same validator.
\item A set of references to transactions $t$ the validator signs.
\end{itemize}
The ack is signed by the validator's secret key.
\end{definition}

All messages can only reference previously created messages with hashes. Cyclic hash references are impossible and hence all messages form a directed acyclic graph (DAG), with the genesis being the only root. Messages are processed in any order respecting references. Agents do not process a transaction $t$ until $\past(t)$ is received in full.

\begin{definition}[past]\label{def:past}
The set of messages reachable by following references from $t$ is called $\past(t)$. For a set of messages $T$, $\past(T) = \bigcup_{t \in T} \past(t)$.
\end{definition}

Transactions can be confirmed by the system, and confirmation is permanent. A transaction $t$ becomes confirmed when enough validators broadcast an ack signing it. After a transaction is confirmed, the {\it stake delegated} to the validator indicated in $t$ increases by the value of $t$ (and appropriately decreases for the validators to whom the inputs were delegated). Thus we define transaction confirmation and the stake delegated to a validator inductively (from genesis) with respect to each other.

Genesis is confirmed from the start.

\begin{definition}[delegated stake]\label{def:stake}
Given a set of acks $A$, let $T_C$ be the set of transactions confirmed in $\past(A)$ that indicate $v$ as the validator. The stake delegated to $v$ in $\past(A)$ is equal to the sum of values of outputs of transactions in $T_C$ that are unspent in $\past(A)$.
\end{definition}

\begin{definition}[confirmed]\label{def:confirmed}
A transaction $t$ is confirmed if the transactions that output the inputs of $t$ are confirmed, and there exists a set of acks $A_t$ such that:
\begin{itemize}
\item some validators $v_1, \dots, v_k$ with respective delegated stake $m_1, \dots, m_k$ in $\past(A_t)$ sign $t$, and $\sum_{i=1}^k m_i > \frac{2}{3} M$;
\item no transaction $t' \in \past(A_t)$ shares an input with $t$.
\end{itemize}
\end{definition}

Honest agents do not spend their inputs more than once. Given some honest transaction $t$, all validators can sign $t$ and for no potential $A_t$ will there be a $t' \in \past(A_t)$ that shares an input with $t$. Then it is straightforward to collect validator acks for $t$ and show that $t$ is confirmed.

On the other hand, if some $t'$ sharing inputs with $t$ is present in the transaction DAG, it might be unclear if there is a set $A_t$ exhibiting that $t$ is confirmed. It is only the misbehaving agent's concern to find an appropriate $A_t$ and prove to the recipient of $t$ that $t$ is confirmed.

\begin{figure}
\centering
\begin{tikzpicture}

\node[tx_conf](GENESIS)
 {$\begin{aligned}
   &\to (p_1, 4), v_1\\
   &\to (p_2, 4), v_2\\
   &\to (p_3, 2), v_3
   \end{aligned}$};
 
\node[below of=GENESIS,
  yshift=-0.25cm,
  anchor=center]
  {Genesis};

\node[tx_conf,
right of=GENESIS,
yshift=1cm,
xshift=1.5cm,
anchor=center
](T1)
 {$\begin{aligned}
       (p_1, 4) &\to (p_4, 3)\\
       & \to (p_5, 1)\\
       & v_4
   \end{aligned}$};

\draw (T1.west) edge[-latex, semithick] (GENESIS.east);

\node[tx,
right of=T1,
yshift=0cm,
xshift=2cm,
anchor=center
](T5)
 {$\begin{aligned}
       (p_4, 3) &\to (p_7, 3)\\
       & v_5
   \end{aligned}$};

\draw (T5.west) edge[-latex, semithick] (T1.east);

\node[tx_conf,
right of=GENESIS,
yshift=-2cm,
xshift=1.5cm,
anchor=center
](T2)
 {$\begin{aligned}
       (p_3, 2) &\to (p_6, 2)\\
       & v_2
   \end{aligned}$};

\draw (T2.west) edge[-latex, semithick] (GENESIS.east);

\node[right of=GENESIS,
yshift=2.5cm,
xshift=2.5cm,
anchor=center,
draw=black, minimum size=7mm, inner sep=1pt, fill=gray!5, thick, circle](V1){$v_1$};

\draw (V1) edge[-latex, semithick] (T1.north);

\node[right of=GENESIS,
yshift=-0.75cm,
xshift=2.5cm,
anchor=center,
draw=black, minimum size=7mm, inner sep=1pt, fill=gray!5, thick, circle](V2){$v_2$};

\draw (V2) edge[-latex, semithick] (T1.south);
\draw (V2) edge[-latex, semithick] (T2.north);

\node[right of=V2,
yshift=0cm,
xshift=2cm,
anchor=center,
draw=black, minimum size=7mm, inner sep=1pt, fill=gray!5, thick, circle](V4){$v_4$};

\draw (V4) edge[-latex, semithick] (T2.east);
\draw (V4) edge[-latex, semithick] (T5.south);

\node[below of=T1,
  yshift=-0.25cm,
  anchor=center]
  {$t_1$};
  
\node[below of=T2,
  yshift=0cm,
  anchor=center]
  {$t_2$};

\end{tikzpicture}
\caption{A subview of the transaction DAG from Figure \ref{fig:example_dag}. The set $A_{t_1}$ consisting of the acks of validators $v_1$ and $v_2$ is a proof that $t_1$ is confirmed. The set $A_{t_2}$ consisting of the acks of validators $v_1$, $v_2$ and $v_4$ is a proof that $t_2$ is confirmed.}
\end{figure}
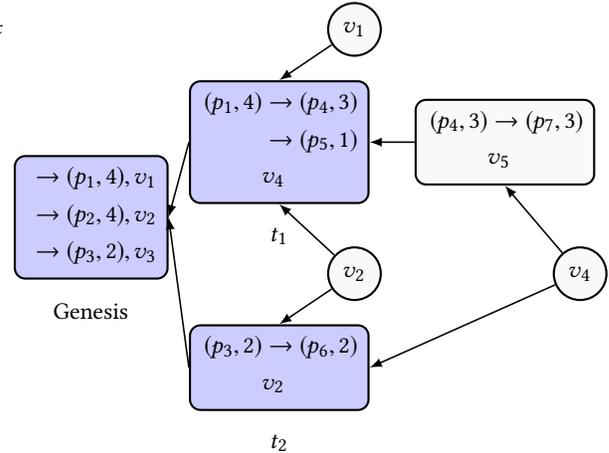

\begin{figure}
\centering
\begin{tikzpicture}

\node[tx_conf](GENESIS)
 {$\begin{aligned}
   &\to (p_1, 3), v_1\\
   &\to (p_2, 2), v_2\\
   &\to (p_3, 2), v_3\\
   &\to (p_4, 2), v_4
   \end{aligned}$};
 
\node[below of=GENESIS,
  yshift=-0.25cm,
  anchor=center]
  {Genesis};

\node[tx_conf,
right of=GENESIS,
yshift=0cm,
xshift=2.5cm,
anchor=center
](T1)
 {$\begin{aligned}
       (p_2, 2) &\to (p_5, 2)\\
       & v_2
   \end{aligned}$};

\draw (T1.west) edge[-latex, semithick] (GENESIS.east);

\node[tx_igno,
right of=GENESIS,
yshift=-3cm,
xshift=2.5cm,
anchor=center
](T2)
 {$\begin{aligned}
       (p_4, 2) &\to (p_7, 2)\\
       & v_5
   \end{aligned}$};

\draw (T2.west) edge[-latex, semithick] (GENESIS.east);

\node[tx_igno,
right of=GENESIS,
yshift=3cm,
xshift=2.5cm,
anchor=center
](T4)
 {$\begin{aligned}
       (p_4, 2) &\to (p_6, 2)\\
       & v_5
   \end{aligned}$};

\draw (T4.west) edge[-latex, semithick] (GENESIS.east);

\node[right of=GENESIS,
yshift=1.5cm,
xshift=3cm,
anchor=center,
draw=black, minimum size=7mm, inner sep=1pt, fill=gray!5, thick, circle](V1){$v_1$};

\draw (V1) edge[-latex, semithick] (T1.north);
\draw (V1) edge[-latex, semithick] (T4.south);

\node[right of=GENESIS,
yshift=-1.5cm,
xshift=3cm,
anchor=center,
draw=black, minimum size=7mm, inner sep=1pt, fill=gray!5, thick, circle](V2){$v_2$};

\draw (V2) edge[-latex, semithick] (T1.south);
\draw (V2) edge[-latex, semithick] (T2.north);

\node[right of=GENESIS,
yshift=-1.5cm,
xshift=5cm,
anchor=center,
draw=black, minimum size=7mm, inner sep=1pt, fill=gray!5, thick, circle](V3){$v_3$};

\draw (V3) edge[-latex, semithick] (T1.south);
\draw (V3) edge[-latex, semithick] (T2.north);

\node[right of=T4,
yshift=0cm,
xshift=1cm,
anchor=center,
draw=black, minimum size=7mm, inner sep=1pt, fill=gray!5, thick, circle](V4){$v_4$};

\draw (V33) edge[-latex, semithick] (T4.east);

\node[right of=T2,
yshift=0cm,
xshift=1cm,
anchor=center,
draw=black, minimum size=7mm, inner sep=1pt, fill=gray!5, thick, circle](V42){$v_4$};

\draw (V42) edge[-latex, semithick] (T2.east);

\end{tikzpicture}

\caption{Example attempt at double-spending. The validator $v_4$ is adversarial, does not reference previous acks in new acks and attempts to confirm conflicting transactions. Honest validators are split between conflicting transactions such that neither will ever be confirmed.}
\end{figure}
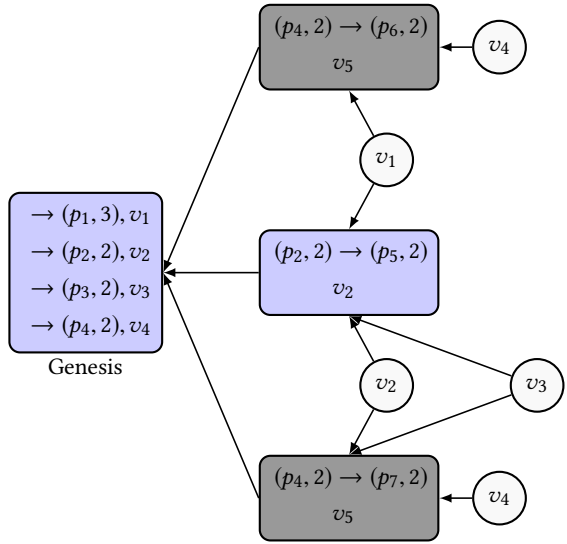

\subsection{Adversary}\label{sec:adversary-stake}

The adversary behaves in an arbitrary way, and thus might create conflicting transactions, acks that do not reference previously issued acks, send different messages to different recipients, etc.

Any message sent by an honest agent is immediately seen by the adversary. The delivery of each message from an honest agent to an honest agent can be delayed by the adversary for an arbitrary amount of time.

\paragraph{\bf Stake.}

We make a standard assumption pertaining to proof-of-stake systems that the adversary does never control more than one-third of the stake. The assumption is the equivalent of assuming that the adversary does not control more than one-third of the permissions in a BFT protocol, or half of the hashing power in a proof-of-work system such as Bitcoin. An agent owning a large stake in a system is heavily invested in the system. While this agent attacking his own system is certainly feasible with deep pockets, it is also self-destructive. So far, larger cryptocurrencies such as Bitcoin have not seen such self-attacks. 

More formally, the value of genesis outputs delegated to the adversary sums up to less than $M/3$. In every transaction, a new validator is indicated, and the stake delegated to the adversary might shift over time.

\begin{definition}[adversary stake]\label{def:valueshift}
Let $m_t$ be the value of a transaction $t$, and $m_{t,A}$ be the sum of values of inputs of $t$ that are outputs of transactions delegated to the adversary. Consequently $m-m_A$ is the sum of values of inputs of $t$ that are outputs of transactions delegated to honest agents.

If a transaction $t$ is delegated to an honest agent, then we subtract $m_{t,A}$ from the amount we count as delegated to the adversary when $t$ is {\em confirmed} (i.e. when some $A_t$ exists).

If transaction $t$ is delegates to the adversary, then we add $m_t - m_{t,A}$ to the amount we count as delegated to the adversary when $t$ is {\em issued}.
\end{definition}

In Section \ref{sec:double-spending} we show that with these assumptions, no conflicting transactions will be confirmed. 

\section{Correctness}\label{sec:double-spending}

This section is devoted to proving that the presented protocol upholds ABC Consensus as defined in Section \ref{sec:relaxing-consensus}. 


Under our assumption from Section \ref{sec:adversary-stake}, more than two-thirds of the money is always delegated to honest validators. Hence, if there is no double-spend alternative to a transaction $t$, honest validators will sign $t$ and $t$ will be confirmed by the system. Thus Validity and Honest-Termination of Definition \ref{def:abcconsensus} hold. Whenever any agent observes a transaction $t$ as confirmed, the acks $A_t$ serve as the proof that $t$ is confirmed to any other agent. Therefore, to show that Agreement holds, it suffices to show that no pair of conflicting transactions is ever confirmed.

\begin{corollary}
If no conflicting transactions are confirmed, ABC protocol satisfies ABC consensus.
\end{corollary}

We now focus on proving that if our assumptions are met, it is impossible that any two conflicting transactions are confirmed, also known as the problem of double-spending.

Every transaction $t$ depends on transactions that happened before $t$ in order for $t$ to be possible.

\begin{definition}[Depends]\label{def:depends}
If a transaction $t'$ spends one or more outputs of transaction $t$, then $t'$ depends on $t$. Dependence is transitive and reflexive, i.e. every transaction depends on itself and if $t_2$ depends on $t_1$ and $t_3$ depends on $t_2$, then $t_3$ depends on $t_1$.
\end{definition}

Any two transactions that together would produce an inconsistent state of the system, such as double-spend transactions, are said to conflict.

\begin{definition}[Conflicts]\label{def:conflicts}
If two transactions $t_1$ and $t_2$ spend the same output, they conflict. Moreover, for two conflicting transactions $t_1$, $t_2$, every transaction that depends on $t_1$ conflicts with every transaction that depends on $t_2$.
\end{definition}

\paragraph{\bf Proof Outline}

For contradiction, assume that some transaction DAG can be produced by the protocol where two conflicting transactions $t_x$ and $t_y$ are confirmed. Consider the instance of such a DAG $G$ that is minimal in terms of the number of transactions.

Although ABC is completely asynchronous, we consider some ordering of messages that represents ``time''. Every message in ABC is ordered such that events in the DAG respect the time, i.e., if an event $t'$ points to an event $t$ in the DAG, the time of $t$ must be smaller than the time of $t'$.

Consider the first transaction $t_0$ that becomes confirmed in DAG $G$ during the protocol's execution. In Lemma \ref{lem:t_in_past_u} we show that for any other confirmed transaction $t$, either $t_0 \in \past(A_t)$ or $t \in \past(A_{t_0})$ holds in DAG $G$ (illustrated in Figure \ref{fig:prop1}). We conclude in Corollary \ref{cor:no_conflict} that $t_0$ cannot conflict with any transaction. In Lemma \ref{lem:acks_the_same} we note that $t_0$ does not serve a purpose for the construction of DAG $G$, as $t_0$'s inputs could be replaced in genesis with $t_0$'s outputs for a smaller DAG. This contradicts with out choice of $G$, and Theorem \ref{thm:no_conficting_confirmed} summarizes that under our assumptions, conflicting transactions cannot be confirmed in a single DAG.

\begin{lemma}\label{lem:depends_confirmed}
If a confirmed transaction $t_2$ depends on $t_1$, $t_1$ is confirmed.
\end{lemma}
\begin{proof}
Follows directly from Definitions \ref{def:confirmed} and \ref{def:depends} by induction.
\end{proof}

\begin{lemma}\label{lem:no_unconfirmed_in_G}
There are no unconfirmed transactions in $G$.
\end{lemma}
\begin{proof}
Suppose some unconfirmed transaction $t_u$ exists in $G$. Since $t_u$ is unconfirmed, by Lemma \ref{lem:depends_confirmed} no confirmed transaction depends on $t_u$.

Let $A$ be some set of acks in DAG $G$. Consider the DAG $G'$ obtained by removing $t_u$ together with transactions that depend on $t_u$, and removing references from acks to $t_u$ (and dependent transactions). Let $A'$ be the set of acks in $G'$ corresponding to $A$. By Definition \ref{def:stake}, for any validator $v$, the stake delegated to $v$ in $A'$ is no less than in $A$. Hence, by Definition \ref{def:confirmed}, any transaction $t$ confirmed in $G$ remains confirmed in $G'$. In particular, $t_x$ and $t_y$ are confirmed in $G'$. However, $G'$ contains less transactions than $G$, a contradiction with $G$ being minimal.
\end{proof}

We argue about the confirmed transactions with respect to some total order of messages in ``time'' that respects message dependence. Of course, between different executions producing the same transaction DAG, some transactions might be confirmed in different orders (or at the same time); we fix one arbitrary possible order of messages $<_{time}$ and choose $t_0$ as the first transaction confirmed in $G$ during the execution, i.e. such that there is a set $A_{t_0}$ where $\max_{<_{time}} A_{t_0}$ is minimal.

\begin{lemma}\label{lem:t_in_past_u}
Let $t$ be some confirmed transaction. Then, either $t_0 \in \past(A_t)$ or $t \in \past(A_{t_0})$ holds in $G$.
\end{lemma}
\begin{proof}
Suppose for contradiction for some $t$ neither $t_0 \in \past(A_t)$ nor $t \in \past(A_{t_0})$ hold in $G$.
Since $t_0$ is the first transaction confirmed during the execution of the protocol, by Definitions \ref{def:stake} and \ref{def:confirmed}, the first possible $A_{t_0}$ contains only acks of validators specified in genesis. The value of genesis outputs delegated to the adversary can only amount to less than $M/3$, so $A_{t_0}$ has to contain honest acks $A_{h,t_0} \subseteq A_{t_0}$ of value larger than $M/3$ (where the genesis outputs associated with $A_{h,t_0}$ are not spent in $\past(A_{t_0})$, by Definition \ref{def:stake}). Let $V_h$ be the (honest) validators that issued acks in $A_{h,t_0}$. Honest validators reference subsequent acks, so since $t \notin \past(A_{t_0})$, $V_h$ have not signed $t$ before signing $t_0$.

We observe:
\begin{itemize}
    \item none of acks in $A_{h,t_0}$ are in $\past(A_t)$;
    \item $V_h$ have not signed $t$ before signing $t_0$;
    \item the genesis outputs delegated to $V_h$ are not spent in $\past(A_{h,t_0})$;
    \item more than $1/3 M$ is delegated to $V_h$ in genesis.
\end{itemize}
Hence, in $\past(A_t)$, any transaction spending the outputs delegated in genesis to $V_h$ can be signed by validators with less than $2/3 M$ stake delegated, and similarly, $t$ can only be signed by validators with less than $2/3 M$ stake delegated in $A_t$, a contradiction.
\end{proof}

\begin{corollary}\label{cor:no_conflict}
There is no transaction conflicting with $t_0$ in $G$.
\end{corollary}
\begin{proof}
Follows from Lemmas \ref{lem:no_unconfirmed_in_G} and \ref{lem:t_in_past_u}.
\end{proof}

\begin{lemma}\label{lem:acks_the_same}
There is a DAG $G'$ where in genesis the inputs of $t_0$ are replaced with the outputs of $t_0$, and the set of confirmed transactions is the same as in $G$, except not including $t_0$.
\end{lemma}
\begin{proof}
Let $V$ be the set of validators to whom the inputs of $t_0$ are delegated in genesis, and $v$ be the validator to whom $t_0$ is delegated. Consider some confirmed transaction $t_1$. By Lemma \ref{lem:t_in_past_u} either $t_1 \in \past(A_{t_0})$ or $t_0 \in \past(A_{t_1})$. Since $t_0$ is spending the outputs delegated to $v$, $v$ can only have signed $t_1$ if $t_0 \notin \past(A_{t_1})$. Then $t_1 \in \past(A_{t_0})$. Hence, for any $A_{t_1}$ where $v$ contributes an ack, we have $t_1 \in \past(A_{t_0})$. If the set $V$ contributed an ack to $A_{t_1}$ and $v$ contributed an ack to $A_{t_2}$, $t_1$ and $t_2$ cannot conflict.

Thus, in $G'$ where $t_0$'s inputs are replaced with $t_0$'s outputs in genesis the validators of the outputs can issue acks equivalent to those in $G$. If inputs and outputs of $t$ do not match in value, genesis can contain smaller outputs that can combine to inputs or outputs of $t_0$ in value.
\end{proof}

\begin{theorem}\label{thm:no_conficting_confirmed}
No DAG can be produced by the ABC protocol such that a pair of confirmed transactions are conflicting.
\end{theorem}
\begin{proof}
Suppose $t_0$ is the first transaction confirmed in $G$ during the execution of the protocol. By Corollary \ref{cor:no_conflict}, no transaction in $G$ conflicts with $t_0$. By Lemma \ref{lem:acks_the_same}, there is a DAG $G'$ containing fewer transactions than $G$ with some pair $t_x$, $t_y$ of conflicting transactions confirmed, a contradiction with minimality of $G$.
\end{proof}

\begin{corollary}
ABC protocol satisfies ABC Consensus.
\end{corollary}

\begin{figure}
\centering
\begin{tikzpicture}

\node[draw=black, minimum size=7mm, inner sep=1pt, fill=blue!20, thick, rectangle](GENESIS){};
 
\node[below of=GENESIS,
  yshift=-0.25cm,
  anchor=center]
  {Genesis};

\node[draw=black, minimum size=7mm, inner sep=1pt, thick, rectangle, fill=blue!20,
right of=GENESIS,
yshift=0.5cm,
xshift=0.5cm,
anchor=center](T1){};

\draw (T1) edge[-latex, semithick] (GENESIS);

\node[draw=black, minimum size=7mm, inner sep=1pt, thick, rectangle, fill=blue!20,
right of=GENESIS,
yshift=-0.5cm,
xshift=0.5cm,
anchor=center](T2){};

\draw (T2) edge[-latex, semithick] (GENESIS);

\node[right of=T1,
yshift=-0.5cm,
xshift=0.5cm,
anchor=center](DOTS){\LARGE $\dots$};

\node[above of=DOTS,
yshift=0.5cm,
anchor=center](ADOTS){};

\node[below of=DOTS,
yshift=-0.5cm,
anchor=center](BDOTS){};

\node[right of=DOTS, fill=blue!20,
yshift=-0.75cm,
xshift=0.5cm,
anchor=center,
draw=black, minimum size=7mm, inner sep=1pt, thick, rectangle](V){$t$};

\draw (V) edge[-latex, semithick] ([yshift=-2mm] DOTS.east);
\draw (V) edge[-latex, semithick] ([yshift=3mm] BDOTS.east);

\node[right of=V,
yshift=0.6cm,
xshift=0.5cm,
anchor=center,
draw=black, minimum size=7mm, inner sep=1pt, fill=gray!5, thick, circle](V1){$v_1$};

\draw (V1) edge[-latex, semithick] (V);
\draw (V1) edge[-latex, semithick] (DOTS.east);

\node[right of=V,
yshift=-0.5cm,
xshift=0.5cm,
anchor=center,
draw=black, minimum size=7mm, inner sep=1pt, fill=gray!5, thick, circle](V2){$v_2$};

\draw (V2) edge[-latex, semithick] (V);
\draw (V2) edge[-latex, semithick] (BDOTS.east);


\node[right of=DOTS, fill=blue!20,
yshift=1cm,
xshift=1cm,
anchor=center,
draw=black, minimum size=7mm, inner sep=1pt, thick, rectangle](U){$t'$};

\draw (U) edge[-latex, semithick] ([yshift=2mm] DOTS.east);
\draw (U) edge[-latex, semithick] ([yshift=-3mm] ADOTS.east);

\node[right of=U,
yshift=0.6cm,
xshift=0.5cm,
anchor=center,
draw=black, minimum size=7mm, inner sep=1pt, fill=gray!5, thick, circle](U1){$v_1'$};

\draw (U1) edge[-latex, semithick] (U.east);
\draw (U1) edge[-latex, semithick] (ADOTS.east);

\node[right of=U,
yshift=-0.5cm,
xshift=1.5cm,
anchor=center,
draw=black, minimum size=7mm, inner sep=1pt, fill=gray!5, thick, circle](U2){$v_2'$};

\draw (U2) edge[-latex, semithick] (U);
\draw (U2) edge[-latex, semithick] (V1);

\end{tikzpicture}

\caption{Example illustration of $t \in \past(A_{t'})$. Nodes $v_i$ are acks in $A_t$ and $v_i'$ are acks in $A_{t'}$. If $t$ and $t'$ are confirmed, then $t \in \past(A_{t'})$ or $t' \in \past(A_t)$. Then, $t$ and $t'$ cannot conflict.}
\label{fig:prop1}
\end{figure}
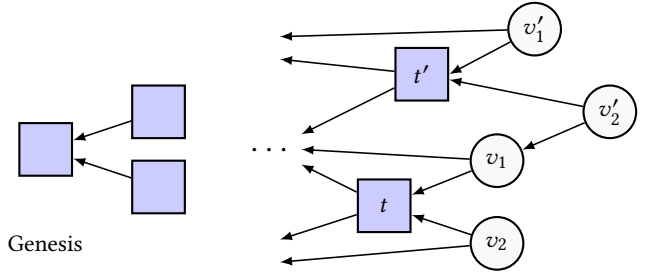

\section{Efficiency and Economy}
\label{sec:improvements}

In this section, we present some practical thoughts that will help ABC succeed.

\subsection{Parallelization}

In some limited scenarios, scaling the throughput by parallel processing is impossible in any system. For example, if all transactions pass the same token in a single chain of dependent transactions, the transactions have to be confirmed in sequence. In ABC, validity of confirmations depends on the set of validators, so similarly, if the set of validators changes dramatically and rapidly all the time, these changes have to be processed sequentially. However, if we rule out these corner cases, validators can parallelize signing and processing of transactions in ABC. Thus, if we increase the number of machines (with constant bandwidth each) at the validator's disposal, the throughput of ABC increases without  limit.

\paragraph{\bf Parallel signing of transactions.} A validator $v$ can split the space of possible input public keys between multiple servers, for example based on the first few characters of the public key. These servers can then individually check that the inputs of a transaction have not been spent already. The transaction needs to be routed only to a small subset of $v$'s machines, which remember and determine whether the validator has signed any transactions spending the same inputs before. Some number of transactions determined to be safe to sign can be combined by the machines in parallel in a Merkle tree root to be included in the next ack in logarithmic time.

\paragraph{\bf Determining transaction confirmation in parallel.} Again, the transaction space can be split between machines that listen to transactions being broadcast in the network. Based on issued and confirmed transactions, each machine can compute the lower and upper bound of how much the stake delegated to each validator changes associated with the assigned transaction space. If transactions $t$ are signed by more validators than the bare operational minimum, the transaction can easily be determined to be confirmed without computing the exact $A_t$, $\past(A_t)$, or the exact delegated stakes associated with validator signatures.

This method is effective in realistic scenarios where more validators than the bare minimum are honest and the transactions do not shift the delegated stake faster than transaction confirmation can be computed.

In some borderline scenarios, computing $A_t$ and $\past(A_t)$ exactly might be unavoidable. For example, suppose the adversary acquires and delegates to itself very close to one-third of the stake, and subsequently refrains from participating in transaction signing. Then all honest validators must sign all transactions, and each transaction might shift the delegated stake such that different combinations of validator signatures confirm a transaction.

\subsection{Pruning the Transaction DAG}

Blockchain systems typically append new transactions to some data structure consisting of previous transactions. In this way, the shared ledger grows over time and requires newly joined participants to process all transactions that took place to date. It is often impossible for the system to completely remove past, intermediate transactions from the data structure and relieve the participants from processing or storing them. Additional protocols might be devised to refer to transactions concisely, such as \cite{kiayias2017non} for proof-of-work blockchains.

We give the outline of how to permanently prune some intermediate transactions from the DAG in ABC. After some time passes and the transaction DAG contains redundant information, a special message called a checkpoint might be created (by any agent) that references a set of acks and transactions $T$ and summarizes the state of the system in $\past(T)$ by listing all confirmed unspent transaction outputs (for example arranged in a Merkle tree) and the distribution of stake among validators in $\past(T)$. Then, some validators that accounted for more than two-thirds of the stake at some point in $\past(P)$ confirm the checkpoint much like they would confirm a normal transaction appended to $\past(P)$.

When the validators listed in the checkpoint issue the next ack after observing the checkpoint, they reference the checkpoint and indicate the summary of transactions outside of $\past(T)$ they have signed up to this point.

Any agents newly joining the system only needs to process $\past(P)$ and the checkpoint with subsequent acks, avoiding processing any spent outputs that took place in $\past(T) \setminus \past(P)$.

For example, a set of validators accounting for more than two-thirds of the stake in the genesis might sign a checkpoint after some time, effectively making the checkpoint a new genesis. An example is illustrated in Figure \ref{fig:checkpoint} in the appendix.

This checkpointing process relies on similar principles as normal transaction confirmation and does not solve consensus. The validators confirm that the checkpoint succinctly summarizes $\past(T)$, but transactions included in the checkpoint might already be spent when the checkpoint is created (which would be revealed by examination of subsequent acks).

\subsection{Transaction Fees}
\label{sec:txfees}

To incentivize maintaining of the system by the validators and prevent spamming attacks, transaction fees should be introduced. ABC is not fundamentally associated with any fee structure in particular and many alternative fee structures are possible. We present an example fee structure.
As acks only serve to confirm transactions, they pay no fee.

Since ABC does not establish consensus, we will refrain from attempting to choose an agreed upon subset of validators that should receive the fee from a particular transaction. Instead, we suggest that all validators are eligible to a portion of every transaction fee. For example, if $m$ money (of the total of $M$) is delegated to a validator $v$, $v$ is granted a fee share $\frac{m}{M} \fee(t)$, where $\fee(t)$ is the transaction fee paid by the issuer of transaction $t$. The fee amount might depend on the size of the representation of $t$. Unlike Bitcoin, the fee cannot easily be determined by a market mechanism, as we need more than two-thirds of the validator power to sign transaction $t$. 

Agents can freely decide whom to choose as a validator. Validators $v$ might have an incentive to reimburse the agents that have delegated their stake to $v$, or the protocol itself might automatically share the transaction fees of $v$ with the money holders. In this way, the transaction fees can flow back to the participants of the ``proof-of-stake pool''. Similarly to mining pools, validators will try to lure as many possible agents into their pool by offering the best conditions.

In Section \ref{sec:validators} we argued that validators may be financially punished if they are too small. This punishment can be incorporated in the fees, e.g., by simply computing the fee to be 0 if the delegated stake of a validator is out of bounds. In case of too small validators, this directly solves the problem.

On the other hand, it is impossible to prevent or mitigate a validator being too large (too much delegated stake), as the validator can split into two validators to stay below the stake threshold. In other words, a large validator can hide behind multiple identities. Note that such a problem is by no means exclusive to ABC. All other blockchain protocols have this issue, e.g. in Bitcoin a mining pool may simply split if it was seen as ``too big''.


\subsection{Money Creation}

If the fees are collected according to the previous section, the fees paid by transaction issuers equal the fees collected by those confirming them. If some transactions are never ack'ed by some individual validators, money will be destroyed, and the total amount of money $M$ may shrink over an extended period of time. However, for economic reasons, we rather might want that the amount of money increases over time. If we introduce 2\% new money each year (as an inflation target), agents will be disincentivized from hoarding their money, but actually use the system. For example, the collected transaction fee might be multiplied by some constant $\alpha > 1$: The issuer is paying $\fee(t)$, but a validator with a delegated stack of $m$ will collect $\alpha \frac{m}{M} \fee(t)$. Assuming every agent possesses less than one-third of the overall stake at any point, issuing transactions still incurs a cost as long as $\alpha \le 3$.

As an additional role in Bitcoin and related systems, proof-of-work serves to distribute newly created money in an unbiased way. ABC could employ proof-of-work for this purpose as well. For this purpose, transactions could be allowed to include proof-of-work and receive an extra amount of stake to spend as an output. However, for these rewards to vary over time, we would need to introduce some mechanism for the protocol to record the passage of time, and we leave that to future work.

\subsection{Smart Contracts}
\label{sec:smartcontracts}

Open smart contracts are not easily possible with ABC. Even an account (output) that two partners Alice$_1$ and Alice$_2$ can spend may cause trouble. If Alice$_1$ and Alice$_2$ concurrently try to buy some goods using two transactions $t_1$ and $t_2$ respectively, they might end up the situation described in Section \ref{sec:relaxing-consensus}.  While the input may have enough value to pay for both goods, issuing both transactions simultaneously may be seen as two conflicting transactions, and as such as a double-spend.

So the two Alices need to be sure that they are not using the same input at the same time. For example, the two Alices can initially set up a transaction that sends the money in the joint account to two new accounts that they control each. Then each Alice can spend the money from her account.

The issue becomes more intriguing with completely open smart contracts that can be called by anybody in the network (for instance, implementing a gambling service). In this case, ABC would need to be augmented with another mechanism on top, ordering transactions for the same smart contract to make sure that concurrent transactions (``double-spends'') for that smart contract do not happen. Are we back to having to implement a full consensus as in Definition \ref{def:consensus}? Yes and no. Clearly such a service needs to totally order all incoming transactions, in other words, deciding which transaction should be presented to ABC first.

However, this ordering overhead is only necessary for completely open smart contracts, and smart contracts can have separate ordering services. Traditional BFT/blockchain protocols totally order \textit{all} transactions. This is the root of all problems, as it introduces an inherent bottleneck in the design of a system.





\section{Related Work}\label{sec:related}

\paragraph{\bf Permissioned systems.} Even though ABC is a permissionless system, it makes sense to compare it to some permissioned systems as well.

Traditionally, distributed ledgers \cite{lamport1998part,castro1999practical} operate with a carefully selected committee of trusted machines. Such systems are called permissioned. The committee repeatedly decides which transactions to accept, using some form of consensus: The committee agrees on a transaction, votes on and commits that transaction, and only then moves forward to agree on the next transaction.

In a work related to ours, Gupta \cite{gupta2016non} proposes a permissioned transaction system that does not rely on consensus. In this design, a static set of validators is designated to confirm transactions. Our concepts of Section \ref{sec:improvements} (parallization, fees, etc.) do work in the permissioned setting as well, and could be applied to this work.

The authors of \cite{guerraoui2019consensus} show that the consensus number of a Bitcoin-like cryptocurrency is $1$, or in other words, that consensus is not needed. The paper provides an analysis and discussion of which applications rely on consensus and to what extent, all of which is directly relevant to ABC. The authors also argue that parallels can be drawn between a permissioned transaction system and the problem of reliable broadcast \cite{malkhi1997secure}.

HoneyBadger BFT \cite{miller2016honey} provides an asynchronous permissioned system by relying on advanced cryptographic techniques with full consensus. Again, the main differences from ABC are that the system is permissioned, much more involved, and reliant on randomization.

The authors of \cite{guerraoui2020dynamic} introduce a protocol based on reliable broadcast that allows participants to join and leave the system. In contrast to ABC, the protocol consists of multiple rounds of communication to agree on nodes joining or leaving the system and does not feature a functionality to delegate one's role in maintaining the system. Node communication volume increases with the number of participants, therefore it cannot be applied in permissionless contexts.

\paragraph{\bf Permissionless systems.} Bitcoin \cite{nakamoto2008bitcoin} radically departed from the established model and became the first permissionless blockchain. In the Bitcoin system, there is no fixed committee; instead, everybody can participate. Bitcoin achieves this by using proof-of-work. Proof-of-work is a randomized process tying computational power and spent energy to the system's security, while also requiring synchronous communication. However, Bitcoin's form of consensus hardly satisfies the traditional consensus definition. Instead of terminating at any point, the extent to which the consensus is ensured raises over time, approaching but never reaching certainty. More precisely, in Bitcoin transactions are never finalized, and can be reverted with ever decreasing probability. 

Similarly to Bitcoin, ABC allows permissionless participation. In contrast to Bitcoin, ABC does not rely on proof-of-work or randomization, features parallelizability and finality, and works under full asynchrony.

To address the problems associated with proof-of-work, proof-of-stake has been suggested, first in a discussion on an online forum \cite{proofofstakepost}. Proof-of-stake blockchains are managed by participants holding a divisible and transferable digital resource, as opposed to holding hardware and spending energy. Academic works proposing proof-of-stake systems include designs such as Ouroboros \cite{kiayias2017ouroboros} or Algorand \cite{gilad2017algorand}. Proof-of-stake blockchains solve consensus and thus do not parallelize without compromises. The reliance on synchronous communication and randomization in proof-of-stake are potential security risks. Despite avoiding these pitfalls, ABC is also simpler.

\paragraph{\bf DAG blockchains.}
To increase the relatively modest throughput of Bitcoin, some proof-of-work protocols employ directed acyclic graphs in the place of Bitcoin's single chain. SPECTRE \cite{sompolinsky2016spectre} is likely the closest relative of ABC among such protocols, as it relaxes consensus similarly to ABC. However, the similarities are largely superficial, as SPECTRE remains a proof-of-work protocol, employs different techniques, and does not share the other of ABC's advantages. SPECTRE improves many aspects of Bitcoin, but with respect to the harsh criteria of Table \ref{tbl:comparison}, SPECTRE can only earn a tick at permissionless.


\section{Conclusions}

In this paper we presented ABC, an asynchronous blockchain without consensus.
ABC provides the functionality of Bitcoin without consensus, without proof-of-work, without requiring synchronous communication, without relying on randomness. ABC is scalable and with finality. The design of ABC is arguably the simplest possible design for a variety of blockchain applications.

ABC provides an advantageous solution for applications like cryptocurrencies, where honest participants do not generate conflicting status updates. However, a general smart contracts platform like Ethereum requires full consensus. Implementing open smart contracts is not impossible with ABC, but it would need another layer of indirection as sketched in Section \ref{sec:smartcontracts}. Adding this extra layer would  check the last box in Table \ref{tbl:comparison}.


\bibliographystyle{ACM-Reference-Format}
\bibliography{refs}


\begin{thebibliography}{14}


\ifx \showCODEN    \undefined \def \showCODEN     #1{\unskip}     \fi
\ifx \showDOI      \undefined \def \showDOI       #1{#1}\fi
\ifx \showISBNx    \undefined \def \showISBNx     #1{\unskip}     \fi
\ifx \showISBNxiii \undefined \def \showISBNxiii  #1{\unskip}     \fi
\ifx \showISSN     \undefined \def \showISSN      #1{\unskip}     \fi
\ifx \showLCCN     \undefined \def \showLCCN      #1{\unskip}     \fi
\ifx \shownote     \undefined \def \shownote      #1{#1}          \fi
\ifx \showarticletitle \undefined \def \showarticletitle #1{#1}   \fi
\ifx \showURL      \undefined \def \showURL       {\relax}        \fi
\providecommand\bibfield[2]{#2}
\providecommand\bibinfo[2]{#2}
\providecommand\natexlab[1]{#1}
\providecommand\showeprint[2][]{arXiv:#2}

\bibitem[\protect\citeauthoryear{Castro, Liskov, et~al\mbox{.}}{Castro
  et~al\mbox{.}}{1999}]%
        {castro1999practical}
\bibfield{author}{\bibinfo{person}{Miguel Castro}, \bibinfo{person}{Barbara
  Liskov}, {et~al\mbox{.}}} \bibinfo{year}{1999}\natexlab{}.
\newblock \showarticletitle{Practical Byzantine fault tolerance}. In
  \bibinfo{booktitle}{\emph{OSDI}}, Vol.~\bibinfo{volume}{99}.
  \bibinfo{pages}{173--186}.
\newblock


\bibitem[\protect\citeauthoryear{Gilad, Hemo, Micali, Vlachos, and
  Zeldovich}{Gilad et~al\mbox{.}}{2017}]%
        {gilad2017algorand}
\bibfield{author}{\bibinfo{person}{Yossi Gilad}, \bibinfo{person}{Rotem Hemo},
  \bibinfo{person}{Silvio Micali}, \bibinfo{person}{Georgios Vlachos}, {and}
  \bibinfo{person}{Nickolai Zeldovich}.} \bibinfo{year}{2017}\natexlab{}.
\newblock \showarticletitle{Algorand: Scaling byzantine agreements for
  cryptocurrencies}. In \bibinfo{booktitle}{\emph{Proceedings of the 26th
  Symposium on Operating Systems Principles}}. ACM, \bibinfo{pages}{51--68}.
\newblock


\bibitem[\protect\citeauthoryear{Guerraoui, Komatovic, and
  Seredinschi}{Guerraoui et~al\mbox{.}}{2020}]%
        {guerraoui2020dynamic}
\bibfield{author}{\bibinfo{person}{Rachid Guerraoui}, \bibinfo{person}{Jovan
  Komatovic}, {and} \bibinfo{person}{Dragos-Adrian Seredinschi}.}
  \bibinfo{year}{2020}\natexlab{}.
\newblock \showarticletitle{Dynamic Byzantine Reliable Broadcast [Technical
  Report]}.
\newblock \bibinfo{journal}{\emph{arXiv preprint arXiv:2001.06271}}
  (\bibinfo{year}{2020}).
\newblock


\bibitem[\protect\citeauthoryear{Guerraoui, Kuznetsov, Monti, Pavlovi{\v{c}},
  and Seredinschi}{Guerraoui et~al\mbox{.}}{2019}]%
        {guerraoui2019consensus}
\bibfield{author}{\bibinfo{person}{Rachid Guerraoui}, \bibinfo{person}{Petr
  Kuznetsov}, \bibinfo{person}{Matteo Monti}, \bibinfo{person}{Matej
  Pavlovi{\v{c}}}, {and} \bibinfo{person}{Dragos-Adrian Seredinschi}.}
  \bibinfo{year}{2019}\natexlab{}.
\newblock \showarticletitle{The Consensus Number of a Cryptocurrency}. In
  \bibinfo{booktitle}{\emph{Proceedings of the 2019 ACM Symposium on Principles
  of Distributed Computing}}. ACM, \bibinfo{pages}{307--316}.
\newblock


\bibitem[\protect\citeauthoryear{Gupta}{Gupta}{2016}]%
        {gupta2016non}
\bibfield{author}{\bibinfo{person}{Saurabh Gupta}.}
  \bibinfo{year}{2016}\natexlab{}.
\newblock \emph{\bibinfo{title}{A Non-Consensus Based Decentralized Financial
  Transaction Processing Model with Support for Efficient Auditing}}.
\newblock \bibinfo{thesistype}{Master's\ thesis}.
\newblock


\bibitem[\protect\citeauthoryear{Kiayias, Miller, and Zindros}{Kiayias
  et~al\mbox{.}}{2017a}]%
        {kiayias2017non}
\bibfield{author}{\bibinfo{person}{Aggelos Kiayias}, \bibinfo{person}{Andrew
  Miller}, {and} \bibinfo{person}{Dionysis Zindros}.}
  \bibinfo{year}{2017}\natexlab{a}.
\newblock \showarticletitle{Non-Interactive Proofs of Proof-of-Work.}
\newblock \bibinfo{journal}{\emph{IACR Cryptology ePrint Archive}}
  \bibinfo{volume}{2017}, \bibinfo{number}{963} (\bibinfo{year}{2017}),
  \bibinfo{pages}{1--42}.
\newblock


\bibitem[\protect\citeauthoryear{Kiayias, Russell, David, and
  Oliynykov}{Kiayias et~al\mbox{.}}{2017b}]%
        {kiayias2017ouroboros}
\bibfield{author}{\bibinfo{person}{Aggelos Kiayias}, \bibinfo{person}{Alexander
  Russell}, \bibinfo{person}{Bernardo David}, {and} \bibinfo{person}{Roman
  Oliynykov}.} \bibinfo{year}{2017}\natexlab{b}.
\newblock \showarticletitle{Ouroboros: A provably secure proof-of-stake
  blockchain protocol}. In \bibinfo{booktitle}{\emph{Annual International
  Cryptology Conference}}. Springer, \bibinfo{pages}{357--388}.
\newblock


\bibitem[\protect\citeauthoryear{Lamport}{Lamport}{1998}]%
        {lamport1998part}
\bibfield{author}{\bibinfo{person}{Leslie Lamport}.}
  \bibinfo{year}{1998}\natexlab{}.
\newblock \showarticletitle{The part-time parliament}.
\newblock \bibinfo{journal}{\emph{ACM Transactions on Computer Systems (TOCS)}}
  \bibinfo{volume}{16}, \bibinfo{number}{2} (\bibinfo{year}{1998}),
  \bibinfo{pages}{133--169}.
\newblock


\bibitem[\protect\citeauthoryear{Malkhi, Merritt, and Rodeh}{Malkhi
  et~al\mbox{.}}{1997}]%
        {malkhi1997secure}
\bibfield{author}{\bibinfo{person}{Dahlia Malkhi}, \bibinfo{person}{Michael
  Merritt}, {and} \bibinfo{person}{Ohad Rodeh}.}
  \bibinfo{year}{1997}\natexlab{}.
\newblock \showarticletitle{Secure reliable multicast protocols in a WAN}. In
  \bibinfo{booktitle}{\emph{Proceedings of 17th International Conference on
  Distributed Computing Systems}}. IEEE, \bibinfo{pages}{87--94}.
\newblock


\bibitem[\protect\citeauthoryear{Miller, Xia, Croman, Shi, and Song}{Miller
  et~al\mbox{.}}{2016}]%
        {miller2016honey}
\bibfield{author}{\bibinfo{person}{Andrew Miller}, \bibinfo{person}{Yu Xia},
  \bibinfo{person}{Kyle Croman}, \bibinfo{person}{Elaine Shi}, {and}
  \bibinfo{person}{Dawn Song}.} \bibinfo{year}{2016}\natexlab{}.
\newblock \showarticletitle{The honey badger of BFT protocols}. In
  \bibinfo{booktitle}{\emph{Proceedings of the 2016 ACM SIGSAC Conference on
  Computer and Communications Security}}. ACM, \bibinfo{pages}{31--42}.
\newblock


\bibitem[\protect\citeauthoryear{Nakamoto}{Nakamoto}{2008}]%
        {nakamoto2008bitcoin}
\bibfield{author}{\bibinfo{person}{Satoshi Nakamoto}.}
  \bibinfo{year}{2008}\natexlab{}.
\newblock \showarticletitle{Bitcoin: A peer-to-peer electronic cash system}.
\newblock \bibinfo{howpublished}{\url{http://bitcoin.org/bitcoin.pdf}}.
\newblock  (\bibinfo{year}{2008}).
\newblock


\bibitem[\protect\citeauthoryear{QuantumMechanic}{QuantumMechanic}{2011}]%
        {proofofstakepost}
\bibfield{author}{\bibinfo{person}{QuantumMechanic}.}
  \bibinfo{year}{2011}\natexlab{}.
\newblock
  \bibinfo{howpublished}{\url{https://bitcointalk.org/index.php?topic=27787.0}}.
\newblock


\bibitem[\protect\citeauthoryear{Sompolinsky, Lewenberg, and Zohar}{Sompolinsky
  et~al\mbox{.}}{2016}]%
        {sompolinsky2016spectre}
\bibfield{author}{\bibinfo{person}{Yonatan Sompolinsky}, \bibinfo{person}{Yoad
  Lewenberg}, {and} \bibinfo{person}{Aviv Zohar}.}
  \bibinfo{year}{2016}\natexlab{}.
\newblock \showarticletitle{SPECTRE: A Fast and Scalable Cryptocurrency
  Protocol.}
\newblock \bibinfo{journal}{\emph{IACR Cryptology ePrint Archive}}
  \bibinfo{volume}{2016} (\bibinfo{year}{2016}), \bibinfo{pages}{1159}.
\newblock


\bibitem[\protect\citeauthoryear{Wood et~al\mbox{.}}{Wood
  et~al\mbox{.}}{2014}]%
        {wood2014ethereum}
\bibfield{author}{\bibinfo{person}{Gavin Wood} {et~al\mbox{.}}}
  \bibinfo{year}{2014}\natexlab{}.
\newblock \showarticletitle{Ethereum: A secure decentralised generalised
  transaction ledger}.
\newblock \bibinfo{journal}{\emph{Ethereum project yellow paper}}
  \bibinfo{volume}{151}, \bibinfo{number}{2014} (\bibinfo{year}{2014}),
  \bibinfo{pages}{1--32}.
\newblock


\end{thebibliography}

\section*{Appendix}

\appendix

\paragraph{\bf Asynchrony.}

In Table \ref{tbl:comparison}, we call many protocols not asynchronous. In this section, we quickly want to justify this classification.

For many blockchain protocols such as Bitcoin, the underlying network being asynchronous would be devastating. The adversary could simply split the network in two, half of the agents on one side, and half of the agents on the other side. Then the adversary can double-spend all its money on both sides. Since the two sides cannot communicate due to temporarily not receiving any messages from the other side, both sides will eventually have the double-spending transaction in their branch of the blockchain. The blocks containing the transactions will eventually be confirmed by enough (e.g., six) blocks, and the transactions are considered final by merchants. By controlling communication, the adversary can double-spend its money.

One may think that BFT protocol such as PBFT \cite{castro1999practical} can handle asynchrony better. To some extent this is true, as PBFT will not allow such a double-spend, since PBFT and similar protocols provide safety even in asynchronous networks. However, asynchrony is still a problem for BFT protocols such as PBFT, as no more progress (liveness) can be made. PBFT and other protocols handle this issue by adopting a semi-synchronous model which is increasing the time limits whenever messages do not arrive in time. This may dramatically slow down the protocol, as a byzantine agent can simply wait with sending messages before timers run out. 

ABC on the other hand does not have to deal with timing assumptions: Whenever a message arrives, the system makes progress towards establishing or confirming a transaction. Few systems, such as HoneyBadger BFT or (consensus-less) broadcast-based protocols, share this resiliency to asynchrony with ABC.

\begin{figure*}
\centering
\begin{subfigure}{\textwidth}
\centering
\begin{tikzpicture}

\node[tx_conf](GENESIS)
 {$\begin{aligned}
   &\to (p_1, 1), v_1\\
   &\to (p_2, 3), v_2
   \end{aligned}$};
 
\node[below of=GENESIS,
  yshift=-0.1cm,
  anchor=center]
  {Genesis};

\node[tx_conf,
right of=GENESIS,
yshift=0cm,
xshift=2.5cm,
anchor=center
](T1)
 {$\begin{aligned}
       (p_2, 3) &\to (p_3, 3)\\
       & v_3
   \end{aligned}$};

\draw (T1.west) edge[-latex, semithick] (GENESIS.east);

\node[tx_conf,
right of=T1,
yshift=0cm,
xshift=3cm,
anchor=center
](T5)
 {$\begin{aligned}
       (p_3, 3) &\to (p_4, 3)\\
       & v_4
   \end{aligned}$};

\draw (T5.west) edge[-latex, semithick] (T1.east);

\node[tx_conf,
right of=T5,
yshift=-1.5cm,
xshift=2.5cm,
anchor=center
](T3)
 {$\begin{aligned}
       &\to (p_4, 3), v_4\\
       &\to (p_1, 1), v_1
   \end{aligned}$};

\node[right of=GENESIS,
yshift=1.5cm,
xshift=3cm,
anchor=center,
draw=black, minimum size=7mm, inner sep=1pt, fill=gray!5, thick, circle](V2){$v_2$};

\node[right of=V2,
yshift=0cm,
xshift=3.25cm,
anchor=center,
draw=black, minimum size=7mm, inner sep=1pt, fill=gray!5, thick, circle](V3){$v_3$};

\draw (V3) edge[-latex, semithick] (T5.north);

\draw (V2) edge[-latex, semithick] (T1.north);
\draw (V3) edge[-latex, semithick] (V2);
\draw (T3.west) edge[-latex, semithick] (V3);
\draw (T3.west) edge[-latex, semithick] (GENESIS.south);

\node[above of=T3,
  yshift=0cm,
  anchor=center]
  {Checkpoint};

\node[right of=T3,
yshift=0cm,
xshift=1cm,
anchor=center,
draw=black, minimum size=7mm, inner sep=1pt, fill=gray!5, thick, circle](V22){$v_2$};

\draw (V22) edge[-latex, semithick] (T3.east);
\end{tikzpicture}

\caption{Example transaction DAG. Some agent issued a checkpoint and the validator $v_2$ signs the checkpoint as accurate.}
\end{subfigure}

\vspace{2em}

\begin{subfigure}{\textwidth}
\centering
\begin{tikzpicture}

\node[tx_conf](GENESIS)
 {$\begin{aligned}
   &\to (p_1, 1), v_1\\
   &\to (p_2, 3), v_2
   \end{aligned}$};
 
\node[below of=GENESIS,
  yshift=-0.1cm,
  anchor=center]
  {Genesis};

\node[tx_conf,
right of=GENESIS,
yshift=0cm,
xshift=2.5cm,
anchor=center
](T1)
 {$\begin{aligned}
       (p_2, 3) &\to (p_3, 3)\\
       & v_3
   \end{aligned}$};

\draw (T1.west) edge[-latex, semithick] (GENESIS.east);

\node[tx_conf,
right of=T1,
yshift=0cm,
xshift=3cm,
anchor=center
](T5)
 {$\begin{aligned}
       (p_3, 3) &\to (p_4, 3)\\
       & v_4
   \end{aligned}$};

\draw (T5.west) edge[-latex, semithick] (T1.east);

\node[tx,
right of=GENESIS,
yshift=3cm,
xshift=2.5cm,
anchor=center
](T4)
 {$\begin{aligned}
       (p_1, 1) \to (p_5, 1)
   \end{aligned}$};

\draw (T4.west) edge[-latex, semithick] (GENESIS.east);

\node[tx_conf,
right of=T5,
yshift=-1.5cm,
xshift=2.5cm,
anchor=center
](T3)
 {$\begin{aligned}
       &\to (p_4, 3), v_4\\
       &\to (p_1, 1), v_1
   \end{aligned}$};

\node[right of=GENESIS,
yshift=1.5cm,
xshift=3cm,
anchor=center,
draw=black, minimum size=7mm, inner sep=1pt, fill=gray!5, thick, circle](V2){$v_2$};

\node[right of=V2,
yshift=0cm,
xshift=3.25cm,
anchor=center,
draw=black, minimum size=7mm, inner sep=1pt, fill=gray!5, thick, circle](V3){$v_3$};

\node[right of=T4,
yshift=0cm,
xshift=5cm,
anchor=center,
draw=black, minimum size=7mm, inner sep=1pt, fill=gray!5, thick, circle](V4){$v_4$};

\draw (V3) edge[-latex, semithick] (T5.north);

\draw (V2) edge[-latex, semithick] (T1.north);
\draw (V4) edge[-latex, semithick] (T4.east);
\draw (V3) edge[-latex, semithick] (V2);
\draw (T3.west) edge[-latex, semithick] (V3);
\draw (V4) edge[-latex, semithick] (V3);
\draw (T3.west) edge[-latex, semithick] (GENESIS.south);

\node[above of=T3,
  yshift=0cm,
  anchor=center]
  {Checkpoint};

\node[right of=T3,
yshift=0cm,
xshift=1cm,
anchor=center,
draw=black, minimum size=7mm, inner sep=1pt, fill=gray!5, thick, circle](V22){$v_2$};

\draw (V22) edge[-latex, semithick] (T3.east);

\node[right of=T3,
yshift=2cm,
xshift=1cm,
anchor=center,
draw=black, minimum size=7mm, inner sep=1pt, fill=gray!5, thick, circle](V44){
$\begin{aligned}
		&v_4:\\
		(p_1, 1) &\to (p_5, 1)
   \end{aligned}$
};

\draw (V44) edge[-latex, semithick] (V4);
\draw (V44) edge[-latex, semithick] (T3.east);

\end{tikzpicture}

\caption{The validator $v_4$ repeats its' ack that was not included in the checkpoint.}
\end{subfigure}

\vspace{2em}

\begin{subfigure}{\textwidth}
\centering
\begin{tikzpicture}

\node[tx_conf](GENESIS)
 {$\begin{aligned}
   &\to (p_1, 1), v_1\\
   &\to (p_2, 3), v_2
   \end{aligned}$};
 
\node[below of=GENESIS,
  yshift=-0.1cm,
  anchor=center]
  {Genesis};

\node[tx_conf,
right of=GENESIS,
yshift=0cm,
xshift=6cm,
anchor=center
](T3)
 {$\begin{aligned}
       &\to (p_4, 3), v_4\\
       &\to (p_1, 1), v_1
   \end{aligned}$};

\draw (T3.west) edge[-latex, semithick] (GENESIS.east);

\node[above of=T3,
  yshift=0cm,
  anchor=center]
  {Checkpoint};

\node[right of=T3,
yshift=0cm,
xshift=1cm,
anchor=center,
draw=black, minimum size=7mm, inner sep=1pt, fill=gray!5, thick, circle](V22){$v_2$};

\draw (V22) edge[-latex, semithick] (T3.east);

\node[right of=T3,
yshift=2cm,
xshift=1cm,
anchor=center,
draw=black, minimum size=7mm, inner sep=1pt, fill=gray!5, thick, circle](V44){
$\begin{aligned}
		&v_4:\\
		(p_1, 1) &\to (p_5, 1)
   \end{aligned}$
};

\draw (V44) edge[-latex, semithick] (T3.east);

\end{tikzpicture}

\caption{Agents joining the system do not need to process the transactions summarized by the checkpoint.}
\end{subfigure}

\caption{Example illustration of a checkpoint.}

\label{fig:checkpoint}
\end{figure*}
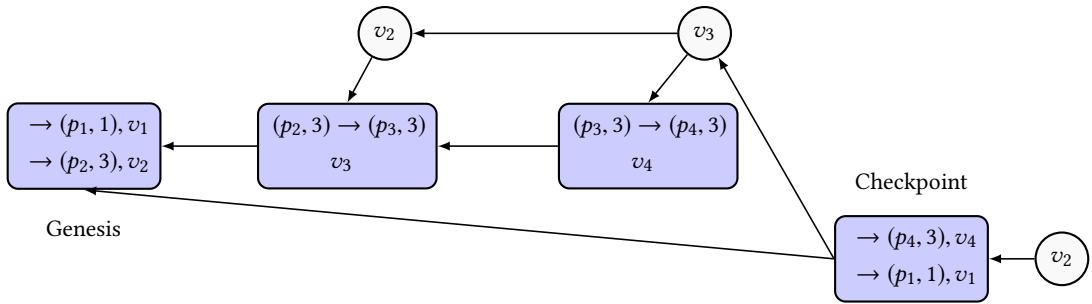
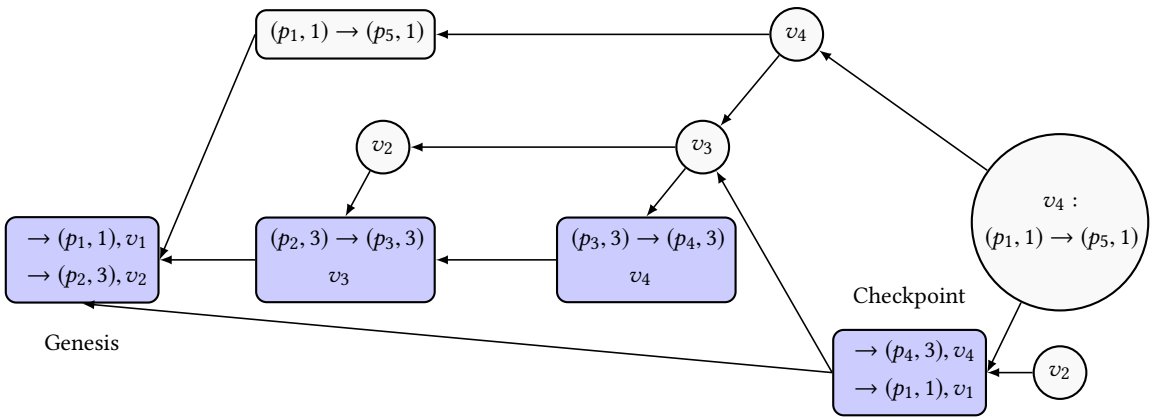
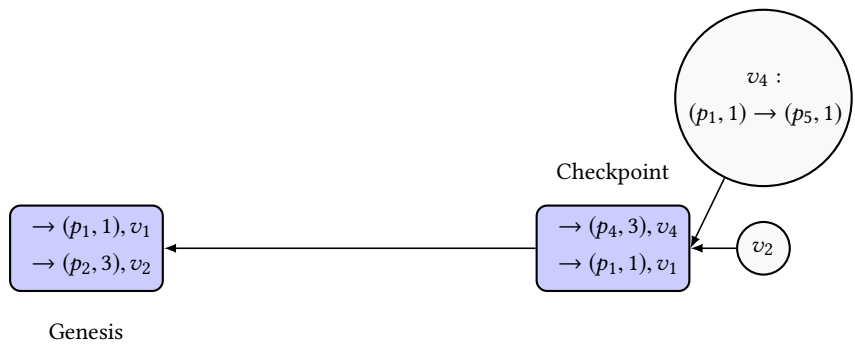

\end{document}